\documentclass[letterpaper,10 pt,conference]{ieeeconf_hyper}

\IEEEoverridecommandlockouts             
\usepackage{cite,url,inputenc}
\usepackage{amsmath,amssymb,amsfonts,balance}
\usepackage{graphicx,epsfig,float}
\usepackage{psfrag}
\usepackage{textcomp}
\usepackage{rotate}
\usepackage{color}
\usepackage{mathptmx} 
\usepackage{times} 

\usepackage[colorlinks=true, linkcolor=blue, citecolor=green, urlcolor=red]{hyperref}

\newcommand{\Rf}{{\mathbb R}}

\newcommand{\Bf}{{\mathbb B}}
\newcommand{\Kn}{{\mathbb K}}

\newcommand{\bls}{\vspace{\baselineskip}}

\newtheorem{theorem}{Theorem}

\newtheorem{remark}{Remark}

\title{\LARGE \bf Optimal and Robust Multivariable Reaching Time Sliding Mode Control Design*}

\author{Jos\'e C. Geromel$^{1}$, Liu Hsu$^{2}$ and Eduardo V. L. Nunes$^{3}$%
\thanks{*This work was supported in part by the National Council for Scientific and Technological Development (CNPq / Brazil), under Grant 302013/2019-9 and Coordination for the Improvement of Higher Education Personnel) (Capes / Brazil), under Grant 001.}
\thanks{$^{1}$Jos\'e C. Geromel is with the School of Electrical and Computer Engineering, UNICAMP, SP-Brazil {\tt\small geromel@dsce.fee.unicamp.br}}%
\thanks{$^{2}$Liu Hsu is with the Department of Electrical Engineering, COPPE, Federal University of Rio de Janeiro, RJ-Brazil {\tt\small lhsu@coppe.ufrj.br}}%
\thanks{$^{3}$Eduardo V. L. Nunes is with the Department of Electrical Engineering, COPPE, Federal University of Rio de Janeiro, RJ-Brazil {\tt\small eduardo.nunes@coppe.ufrj.br}}
}

\begin{document}

\maketitle
\thispagestyle{empty}
\pagestyle{empty}

\begin{abstract}


This paper addresses two minimum reaching time control problems within the context of finite stable systems. The well-known Variable Structure Control (VSC) and Unity Vector Control (UVC) strategies are analyzed, with the primary objective of designing optimal and robust state feedback gains that ensure minimum finite time convergence to the origin. This is achieved in the presence of convex bounded parameter uncertainty and norm-bounded exogenous disturbances. In both cases, the optimality conditions are expressed through Linear Matrix Inequalities (LMIs), which are solved efficiently within the framework of multivariable systems using existing numerical tools. The theoretical results are demonstrated with two practically motivated examples.

\end{abstract}

\section{Introduction}

This paper is a natural follow-up of the recent paper \cite{gero:2023b} that has assessed the reaching time convergence towards the origin of the MSTA control strategy. The LMI-based control design conditions have been established in the former paper \cite{gero:2023}, taking into account robustness against convex bounded parameter uncertainty and exogenous norm bounded disturbance. In the context of finite time convergence and the ability of disturbance rejection of the STA control, the paper \cite{Utkin_tac:2013} presents results of theoretical and practical importance. 

The seminal paper \cite{U:83} offers to the reader a {\it tour d'horizon} about Variable Structure Systems including possible theoretical results that should be developed in the future. This is precisely what we want to present afterwards, by providing reaching time estimation for two classes of finite time control strategies, putting in evidence optimal and robust state feedback control design.

Our interest falls into two well-known classes of finite time systems, namely Variable Structure Control (VSC) \cite{U:1992,UGS:2009} and Unit Vector Control (UVC) \cite{G:79}. First, in an analysis step, the conditions for finite time stability and reaching time estimation are given. For VSC this first step is accomplished by adopting the Persidiskii type Lyapunov function proposed in \cite{Liu_Hsu:2000} whereas for UVC we have considered a simple quadratic function. Second, in the synthesis step, in both cases, conditions expressed by LMIs (involving all scalar parameters, but one) are obtained. This is a key result as far as multivariable systems are concerned. Indeed, the VSC design problem is shown to be jointly convex being thus 
solved with no difficulty. The UVC design needs a line search procedure to determine the global optimal solution with respect to a scalar variable. Again, this is done without difficulty.

This paper treats for the first time, we believe, the minimum reaching time problem in a general setting, including optimality concerning multivariable state feedback gains and robustness with respect to convex bounded parameter uncertainty and norm bounded exogenous disturbance. Moreover, as normally occurs in minimum time problems, the optimal gain matrix tends to be large in order to 
achieve a small reaching time for  the closed-loop system. This undesirable drawback is circumvented by imposing a control norm bound by means of a convex constraint. This is a fact of major importance in practical applications. Hence, the theoretical results are applied to two different examples with practical appeal, namely a second-order Visual Robotics Servoing and a third-order Underwater Remotely Operated Vehicle. 

The paper is organized as follows. In the next section, the problem to be solved is stated and discussed. In section III the Variable Structure Control model is presented and the reaching time is determined by adopting a Persiddiskii type Lyapunov function. In the next section, the same is done to Unit Vector Control by presenting the model and the results relating to reaching time determination by the adoption of a pure quadratic Lyapunov function. Section V is entirely devoted to presenting and discussing two practical examples. Section VI summarizes the conclusions including some recommendations concerning some problems to be faced in order to improve the presented results mainly those related to reaching time optimization in the context of VSC and UVC control design. 
	
\subsection{Notation}
As usual, the symbol $\| \cdot \|$ denotes the Euclidean norm for vectors and the corresponding induced norm for real matrices. For symmetric matrices, $(\bullet)$ denotes the symmetric block. The unit simplex $\Lambda \subset \Rf^N$ stands for the set of all nonnegative vectors with the sum of components equal to one. The convex hull of matrices $\{B_i\}_{i \in \Kn}$, with the same dimensions, is denoted by $\text{co} \{B_i\}_{i \in \Kn}$ and is obtained by the convex combination of the extreme matrices $B_i,~i \in \Kn$, where $\Kn=\{1, 2, \cdots, N\}$ with $N \in \mathbb{N} \setminus \{0\}$. The solutions of the discontinuous differential equations are understood in the Filippov sense \cite{F:88}. For symmetric matrices, ${\bf tr}(\cdot)$ and $\lambda_{max}(\cdot)$ denote the trace and the maximum eigenvalue, respectively. Finally, $I_n$ denotes the $n \times n$ identity matrix. 

\section{Problem Statement}
	
Consider an uncertain multivariable system described by
\begin{align}
\label{eq01}
\dot \sigma = B u + f(t)
\end{align}
where $\sigma \in \mathbb{R}^n$ is the state vector, $u \in \mathbb{R}^m$ is the system input, $\sigma(0)=\sigma_0 \in \Rf^n$ is the initial condition, $B \in \Rf^{n \times m}$ is the uncertain input matrix, and $f(t): \Rf_+ \to \Rf^n$ is the exogenous disturbance. The control law is supposed to be of the general form
\begin{align}
\label{eq02}
u = K \varphi(\sigma) 
\end{align}
where $\varphi : \Rf^n \rightarrow \Rf^n $ is a given function and $K \in \mathbb{R}^{m \times n}$ is the control gain. Moreover, it is assumed that:
\begin{itemize}
\item[i)] The input matrix $B$ has rank $n \leq m$ and $B \in {\mathbb{B}} = { \text{co}}\{B_i\}_{i \in \mathbb{K}}$ with $B_i \in \Rf^{n \times m}, \forall i \in \Kn$. This means that ${\rm rank}(B) = n, \forall B \in \Bf$ and that $B$ is a generic element of $\Bf$, a compact set generated by the convex combinations of the extreme matrices $\{B_i\}_{i \in \mathbb{K}}$. For each $B \in \Bf$, there exists a constant vector $\lambda \in \Lambda$ such that $B = \sum_{i \in \Kn} \lambda_i B_i$.
\item[ii)] The disturbance $f(t)$ is a Lipschitz continuous function. It is not exactly know but its is know that it satisfies the norm bound
\begin{align}
\label{eq03}
\| f(t) \| \leq \delta, \forall t \geq 0
\end{align}
where $\delta$ is a positive constant.
\end{itemize} 

In this paper, our main goal is to design the state feedback matrix gain $K \in \mathbb{R}^{m \times n}$ such that the closed-loop system $\dot \sigma = BK \varphi(\sigma) + f $ is globally finite-time stable with minimum reaching time $T_r(\sigma_0)$. To this end, two classes of control synthesis of the form (\ref{eq02}), fully characterized by the function $\varphi$ to be given afterwards, are considered.

\section{Variable Structure Control (VSC) Design}

In this case, treated in the references \cite{U:1992,UGS:2009}, the closed-loop system model is defined by the well known function
\begin{align} \label{eq04}
\varphi(\sigma) = \left [ \begin{array}{c} {\rm sign}(\sigma_1) \\ \vdots \\ {\rm sign}(\sigma_n) \end{array} \right ]
\end{align}
where the scalar-valued ${\rm sign}$ function is defined as ${\rm sign}(\xi)= \xi / |\xi|$ if $\xi \neq 0$ and ${\rm sign}(0)=0$. Clearly, each component $\varphi_j$ is continuous and differentiable whenever $\sigma_j \neq 0$. Following \cite{Liu_Hsu:2000}, let us adopt the Lyapunov function candidate of Persidiskii type
\begin{align} \label{eq05}
v(\sigma) = \sum_{j=1}^n p_j \int_0^{\sigma_j} {\rm sign}(\xi) d \xi = \sum_{j=1}^n p_j |\sigma_j|
\end{align}
with the positive scalars $p_j>0$ for all $j=1, \cdots, n$ composing the diagonal matrix $P_d = {\rm diag}(p_1, \cdots, p_n)>0$, to be determined. As required, this is a positive definite radially unbounded function.  
\begin{theorem} \label {theo01}
Let $B \in \mathbb{B}$ be given. If there exist a diagonal matrix $0 < P_d \in \Rf^{n \times n}$ and positive scalars $\omega \mu > \delta^2$ such that the inequality
\begin{align} \label{eq06}
K^TB^T P_d + P_d BK + \omega I_n + \mu P_d^2 < 0
\end{align}
holds, then the origin $\sigma = 0$ of the closed-loop system is globally finite-time stable for all exogenous disturbances satisfying (\ref{eq03}). In the affirmative case, the reaching time satisfies $T_r(\sigma_0) \leq 2\mu v(\sigma_0)/(\omega \mu - \delta^2)$.
\end{theorem}
\begin{proof}
Consider the $j$-th component of $v$. The one-sided directional derivative has the form \cite{Lasdon_book:1970}
\begin{align} \label{eq07}
D_+|\sigma_j| = \left \{ \begin{array}{ccc} {\rm sign}(\sigma_j) \dot \sigma_j &,& \sigma_j \neq 0 \\ |\dot \sigma_j| &,& \sigma_j = 0 \end{array} \right .
\end{align}
which allows us to determine $D_+v$ by considering two complementary cases. First, assuming that $\sigma \neq 0$ (i.e. all components of $\sigma$ different of zero) we have 
\begin{align} \label{eq08}
D_+v(\sigma) & = \sum_{j=1}^n p_j {\rm sign}(\sigma_j) \dot \sigma_j \nonumber \\
& = {\rm sign}(\sigma)^T P_d \dot \sigma \nonumber \\
& = {\rm sign}(\sigma)^T P_d \big ( B K {\rm sign}(\sigma) + f \big ) 
\end{align} 
which, taking into account that  (by completing the squares)
\begin{align} \label{eq09}
2 {\rm sign}(\sigma)^T P_d f \leq \mu {\rm sign}(\sigma)^T P_d^2 {\rm sign}(\sigma) + \mu^{-1} f^Tf
\end{align} 
and (\ref{eq06}), becomes 
\begin{align} \label{eq10}
D_+v(\sigma) & < (1/2){\rm sign}(\sigma)^T \big ( - \omega I_n - \mu P_d^{2} \big ) {\rm sign}(\sigma) \nonumber \\
& ~~~~ +  {\rm sign}(\sigma)^T P_df \nonumber \\
& \leq -(\omega/2) \|{\rm sign}(\sigma)\|^2 + (1/2 \mu) \|f\|^2 \nonumber \\
& \leq -(1/2) \big ( \omega - \delta^2/\mu \big ) \nonumber \\
& < 0
\end{align} 
where we have taken into account that $\|{\rm sign}(\sigma)\|^2 \geq 1$ and $\omega \mu > \delta^2$. Second, assuming that $\sigma_j=0$, in some time interval then necessarily $\dot \sigma_j =0$ and by consequence the nonzero components make (\ref{eq10}) true. The other possibility $\sigma_j=0$ and $\dot \sigma_j \neq 0$ may occur only in a set of zero measure. Finally, integrating (\ref{eq10}) we obtain the upper bound to the reaching time $T_r(\sigma_0)$. The proof is complete.
\end{proof}

\bls We now proceed by factorizing the inequality (\ref{eq06}) by multiplying it by $\mu/(\omega \mu - \delta^2)>0$ and making the change of variables $P_d \Leftrightarrow \mu P_d/(\omega \mu - \delta^2)$, yielding 
\begin{align} \label{eq11}
K^TB^T P_d + P_d BK + \Big ( 1 + \frac{\delta^2}{\beta} \Big ) I_n  + \beta P_d^2 < 0
\end{align}
where $\beta=\omega \mu  - \delta^2>0$. With respect to the new matrix variable, the reaching time estimation simplifies to $T_r(\sigma_0) \leq 2 v(\sigma_0)$. Finally, it is easily seem that uncertain $B \in \mathbb{B}$ can be considered with no difficulty. Actually, assuming that (\ref{eq11}) holds for $B=B_i, i \in \Kn$, multiplying both sides by $Z_d = P_d^{-1}$, introducing the new matrix variable $Y = K P_d^{-1}$, the Schur Complement provides the equivalent set of LMIs
\begin{align} \label{eq12}
\left [ \begin{array}{ccc} B_iY + Y^T B_i^T + \beta I_n & Z_d & \delta Z_d \\ \bullet  & - I_n & 0 \\ \bullet & \bullet & - \beta I_n\end{array} \right ] < 0,~i \in \Kn
\end{align}

{\bf Reaching time LMI:} Defining the vector $\zeta_0 \in \Rf^n$ with components $\zeta_{0j} = \sqrt{|\sigma_{0j}|}$ for all $j =1, \cdots, n$, we have that $T_r(\sigma_0) \leq 2 \zeta_0'Z_d^{-1} \zeta_0$ and the {\em reaching time LMI}  
\begin{align} \label{eq13}
\left [ \begin{array}{cc} \theta & \zeta_0^T \\ \zeta_0 & Z_d  \end{array} \right ] > 0
\end{align}
allows the determination of the best reaching time estimation, in the present context, provided by a Persidiskii type Lyapunov function, that satisfies $T_r(\sigma_0) \leq 2 \theta$ where
\begin{align} \label{eq14}
\inf_{Z_d, Y, \beta, \theta} \{ \theta : (\ref{eq12})-(\ref{eq13}) \} 
\end{align}
and is imposed by the optimal control matrix gain $K = Y Z_d^{-1}$. An interesting and expected feature of the result reported in Theorem \ref{theo01} is the impact of the disturbance magnitude in the reaching time of the closed-loop system. Indeed, setting $\delta =0$ the exogenous disturbance is eliminated in which case, $\beta>0$ arbitrarily close to zero is the more favorable feasible solution obtained by eliminating the third row and column of the LMIs in (\ref{eq12}). Numerically speaking, the impact on the convex control design problem (\ref{eq14}) is not important because the LMIs become simpler but only a scalar variable is eliminated. 

{\bf Control bounding LMI:} In many instances the optimal solution of (\ref{eq14}) provides a large control gain $K$ associated to a very small reaching time estimation. Fortunately, for a given $\alpha_u >0$, the {\em control bounding LMI} 
\begin{align} \label{eq14bis}
\left [ \begin{array}{cc} \alpha_u^2 I_m & Y \\ Y^T & Z_d \end{array} \right ] > 0 
\end{align}
is equivalent to $Y^TY < \alpha_u^2 Z_d$ which when multiplied both sides by $Z_d^{-1} $ becomes $K^TK < \alpha_u^2 P_d$ and, by consequence, it imposes $\|u(t)\|^2 \leq \alpha_u^2 {\rm sign}(\sigma)^T P_d {\rm sign}(\sigma) \leq \alpha_u^2 {\bf tr}(P_d)$, a constant upper bound on the control magnitude for all $t \geq 0$. For this reason, the convex constraint (\ref{eq14bis}) is to be included in the control design problem (\ref{eq14}). 

\begin{remark} 
It is simple to verify that, for scalar systems ($n=1$) with $BK = - \kappa <0$ such that $\kappa > \delta$, the closed-loop VSC system is finite time stable and the true reaching time is given by $T_*(\sigma_0) = |\sigma_0| / (\kappa - \delta)$. On the other handle, the most favorable $\beta >0$ that minimizes the left hand side of inequality (\ref{eq11}) is $\beta = \delta P_d^{-1}$ and it becomes $-2P_d(\kappa - \delta) + 1 < 0$ which admits the positive solutions $2P_d > 1/(\kappa - \delta)$ provided that $\kappa > \delta$. In this case, the reaching time estimation $T_r(\sigma_0) < 2P_d |\sigma_0|$ determined by the minimum feasible $P_d>0$ equals the true reaching time  $T_*(\sigma_0)$. This is an important property of the Persidiskii type Lyapunov function adopted in this paper for dealing with VSC design since the estimate is not conservative in this simple case. Observe however that $B$ must be known.
\end{remark} 

\section{Unit Vector Control (UVC) Design}

Now, we turn our attention to the Unit Vector Control approach, originally proposed in \cite{G:79,GL:75}. The function $\varphi(\sigma)$ is defined by
\begin{align} \label{eq15}
\varphi(\sigma) = \left ( \frac{1}{\|\sigma\|} \right ) \left [ \begin{array}{c} \sigma_1 \\ \vdots \\ \sigma_n \end{array} \right ]
\end{align}
The first important observation is that a Persidiskii type function does not exist anymore because $\varphi$ is not component-wise decomposable. In order to make possible the use of the results of \cite{gero:2023b}, it seems to be more convenient the adoption of a pure quadratic Lyapunov function candidate of the form $v(\sigma) = \sigma^T P \sigma$ with $0 < P \in \Rf^{n \times n}$. 
\begin{theorem} \label {theo02}
Let $B \in \mathbb{B}$ be given. If there exist symmetric matrices $0 < P \in \Rf^{n \times n}$, $0 < W \in \Rf^{n \times n}$ and positive scalars $\omega$, $\rho$, $\mu$ such that the inequalities
\begin{align} \label{eq16}
K^TB^T P + P BK + \mu^{-1} \delta^2 I_n + \mu P^2 + W < 0
\end{align}
and
\begin{align} \label{eq17}
\left [ \begin{array}{cc} \omega W - \rho P & I_n \\ I_n & \rho I_n \end{array} \right ] > 0 
\end{align}
hold, then the origin $\sigma =0$ of the closed-loop system is globally finite-time stable for all exogenous disturbances satisfying (\ref{eq03}). In the affirmative case, the reaching time satisfies $T_r \leq \omega \sqrt{v(\sigma_0)}$.
\end{theorem}
\begin{proof}
Considering $\sigma \neq 0$, the time derivative of the Lyapunov function yields
\begin{align} \label{eq18}
\dot v(\sigma)  & = 2\sigma^T P \big ( BK \sigma / \|\sigma\| + f \big ) 
\end{align}
which, together with (\ref{eq16}) and the inequality 
\begin{align} \label{eq19}
2 \sigma^TP \big ( \|\sigma\| f) & \leq \mu \sigma^T P^2 \sigma + \mu^{-1} \|\sigma\|^2 f^T f \nonumber \\
 & \leq \mu \sigma^T P^2 \sigma + \mu^{-1} \delta^2 \|\sigma\|^2
\end{align}
lead to 
\begin{align} \label{eq20}
\dot v(\sigma) & = \big ( 1/\|\sigma\| \big ) \Big ( \sigma^T \big ( K^TB^TP + PBK \big ) \sigma  + 2 \sigma^TP \big ( \|\sigma\| f \big ) \Big ) \nonumber \\
& \leq \big ( 1/\|\sigma\| \big ) \sigma^T \Big (  K^TB^TP + PBK + \mu^{-1} \delta^2 I_n + \mu P^2 \Big ) \sigma \nonumber \\
& < - \big ( 1/\|\sigma\| \big ) \sigma^T W \sigma \nonumber \\
& < 0
\end{align}
proving thus global finite-time stability. Finally, inequality (\ref{eq20}), together with the Lyapunov function candidate, imply that 
\begin{align} \label{eq21}
\dot v  & < - \left ( \frac{\sigma^TW \sigma}{\sigma^T \sigma} \frac{\sigma^TW \sigma}{\sigma^TP \sigma} \right )^{1/2} v^{1/2}   \nonumber \\
& \leq -(2/\omega) v^{1/2} 
\end{align}
where the last inequality follows from the result of Lemma $1$ of \cite{gero:2023b}. The time integration of (\ref{eq21}) yields the reaching time estimation $T_r(\sigma_0) \leq \omega \sqrt{v(\sigma_0)}$. The proof is complete.
\end{proof}

\bls {\bf Reaching time LMI}: Multiplying the inequality (\ref{eq16}) by $\omega^2$ and both sides of  (\ref{eq17}) by ${\rm diag}(\sqrt{\omega} I_n, 1/ \sqrt{\omega}I_n)$, the change of variables $(P, W, \mu, \rho) \Leftrightarrow (\omega^2 P, \omega^2 W, \omega^{-2} \mu, \omega^{-1} \rho)$ shows that inequality (\ref{eq16}) evaluated at the vertices of $\mathbb{B}$ and (\ref{eq17}) are equivalent to
\begin{align} \label{eq22}
K^TB^T P + P BK + \mu^{-1} \delta^2 I_n + \mu P^2 + W < 0,~i \in \Kn
\end{align}
and
\begin{align} \label{eq23}
\left [ \begin{array}{cc}  W - \rho P & I_n \\ I_n & \rho I_n \end{array} \right ] > 0
\end{align}
which when solved provide the reaching time estimation $T_r(\sigma_0) \leq \sqrt{v(\sigma_0)}$. The inequality (\ref{eq23}) imposes $W > \rho P + \rho^{-1}I_n$ which, together with (\ref{eq22}), indicates that the best choice for the matrix variable $W$ is arbitrarily close to that lower bound. Hence, replacing it in (\ref{eq22}), multiplying both sides by $P^{-1}>0$ and taking into account the change of variables $(P^{-1}, KP^{-1}) \Leftrightarrow (Z, Y)$, Schur Complements allow us to determine the set of LMIs associated to each vertex of the uncertain set $\mathbb{B}$, that is  
\begin{align} \label{eq24}
\left [ \begin{array}{ccc} B_iY + Y^TB_i^T + \rho Z + \mu I_n & Z & \delta Z \\ \bullet & -\rho I_n & 0 \\ \bullet & \bullet & - \mu I_n \end{array} \right ] < 0,~i \in \Kn
\end{align}
such that whenever feasible, they provide the reaching time estimation $T_r^2(\sigma_0) \leq \sigma_0^T Z^{-1}\sigma_0$. The reaching time LMI
\begin{align} \label{eq26}
\left [ \begin{array}{cc} \theta & \sigma_0^T \\ \sigma_0 & Z  \end{array} \right ] > 0
\end{align}
allows us to express the reaching time estimation as $T_r(\sigma_0) \leq \sqrt{\theta}$. As a consequence, in the present framework, the best reaching time estimation emerges from the solution of the problem
\begin{align} \label{eq27}
\inf_{\rho} \inf_{Z,Y,\mu,\theta} \left \{ \theta : (\ref{eq24})-(\ref{eq26}) \right \} 
\end{align}
which can be solved with no difficulty since the inner problem, for $\rho>0$ fixed, is a joint convex problem with respect to all involved variables and is expressed by LMIs. The outer problem can be solved by line search with respect to the scalar variable $\rho>0$. The optimal gain is $K=YZ^{-1}$. As it will be seen in the next section, this result is very accurate. Once again, the impact of the disturbance magnitude in the reaching time estimation is apparent. As previously indicated, setting $\delta =0$ the exogenous disturbance is eliminated, in which case $\mu>0$ arbitrarily close to zero is the more favorable feasible solution obtained by eliminating the third row and column of the LMIs in (\ref{eq24}). Numerically speaking, the impact on the control design problem (\ref{eq27}) is negligible mainly because its nature remains unchanged, that is, the scalar variable must be handled by a line search procedure.  

\bls{\bf The control bounding LMI:}
As commented before, the optimal solution of (\ref{eq27}) may provide a large control gain $K$ associated to a very small reaching time estimation. Once again, for a given $\alpha_u >0$, the convex constraint 
\begin{align} \label{eq27bis}
\left [ \begin{array}{cc} \alpha_u^2 I_m & Y \\ Y^T & Z \end{array} \right ] > 0 
\end{align}
is equivalent to $Y^TY < \alpha_u^2 Z$ which when multiplied to the right by $Z^{-1} \sigma/\|\sigma\|$ and to the left by its transpose, gives rise to the constant upper bound on the control magnitude $\|u(t)\|^2 \leq \alpha_u^2 \sigma^T P \sigma / \|\sigma\|^2 \leq \alpha_u^2 \lambda_{max}(P)$, valid for all $t \geq 0$. The convex constraint (\ref{eq27bis}) is to be included in the control design problem (\ref{eq27}). 
\begin{remark} 
As before, for scalar systems ($n=1$) with $BK = - \kappa <0$ such that $\kappa > \delta$, the closed-loop UVC system is finite time stable and the true reaching time is given by $T_*(\sigma_0) = |\sigma_0| / (\kappa - \delta)$. On the other handle, the most favorable $\mu>0$ that minimizes the left hand side of inequality (\ref{eq22}) is $\mu = \delta P^{-1}$. Inequality (\ref{eq23}) is equivalent to $W > \rho P + \rho^{-1}$ which indicates that the most favorable choice to $\rho>0$ is $\rho = P^{-1/2}$. Putting these finds together, inequality (\ref{eq22}) becomes $-2(\kappa - \delta) P + 2 P^{1/2}<0$ which admits the positive solutions $P^{1/2} > 1/(\kappa - \delta)$ provided that $\kappa > \delta$. Hence, the reaching time estimation $T_r(\sigma_0) < P^{1/2} |\sigma_0|$ determined by the minimum feasible $P>0$ equals the true reaching-time $T_*(\sigma_0)$, as in the VSC case. This is an unexpected property of the quadratic Lyapunov function adopted in this paper for dealing with UVC design. 
\end{remark} 

\section{Illustrative examples}

The theoretical results raised so far are now used to solve two illustrative examples borrowed from the literature. The control design problems were numerically solved by Matlab-LmiLab routines and the closed-loop system trajectories have been calculated by Matlab-Simulink routines. For the time simulations we have chosen the Euler first-order integration method with fixed-step size $1e-04$.

\subsection{Visual robotics servoing}

This is a second-order example related to a robotics visual servoing problem borrowed from \cite{MGM:2024} that had been proposed in \cite{KNH:2019} whose uncertain model has been developed in \cite{gero:2023}. It consists of designing a robust controller for a planar kinematic manipulator
whose end-effector image position coordinates $(p_x, p_y)$ are provided by an uncalibrated fixed camera with an optical axis orthogonal to the robot workspace plane, such that the closed-loop system exhibits a finite-time convergence behavior. Defining the state variable $\sigma = [ p_x~p_y]^T \in \Rf^2$, the open-loop system subject to the exogenous disturbance $f(t) = \sqrt{2}[ {\rm sin}(5t)~{\rm sin}(2t)]^T \in \Rf^2 \times \Rf_+$, can be written as (\ref{eq01}) where
\begin{align*}
B(\phi) = \left [ \begin{array}{cc} {\rm cos}(\phi) & {\rm sin}(\phi) \\ - {\rm sin}(\phi) & {\rm cos}(\phi) \end{array} \right ]
\end{align*}
is a matrix depending on the uncertain (uncalibrated) rotation angle $\phi$ corresponding to the camera/workspace transformation. For a given nominal angle $\bar \phi$, and $\Delta \phi = \phi - \bar \phi$ such that $|\Delta \phi | \leq \bar \Delta$ then $B(\phi) = B(\Delta \phi) B(\bar \phi)$. At this point, the key observation is that
\begin{align*}
\left [ \begin{array}{c} {\rm cos}(\Delta \phi) \\ {\rm sin}(\Delta \phi) \end{array} \right ]  \in & {\rm co} \left ( \left [ \begin{array}{c} {\rm cos}(\bar \Delta) \\ {\rm sin}(\bar \Delta) \end{array} \right ], \left [ \begin{array}{c} 1 \\ {\rm sin}(\bar \Delta) \end{array} \right ], \right . \nonumber \\
& \left . ~~~~ \left [ \begin{array}{c} {\rm cos}(\bar \Delta) \\ -{\rm sin}(\bar \Delta) \end{array} \right ], \left [ \begin{array}{c} 1 \\ -{\rm sin}(\bar \Delta) \end{array} \right ] \right )
\end{align*}
for all $|\Delta \phi | \leq \bar \Delta$, provided that $0 \leq \bar \Delta \leq \pi/2~[{\rm rad}]$. Taking this fact into account it is simple to determine four ($N=4$) extreme matrices $B_i \in \Rf^{2 \times 2}, i \in \Kn$ such that
\begin{align*}
B(\phi) \in \Bf = {\rm co}(B_i)_{i \in \Kn},~ \forall |\Delta \phi | \leq \bar \Delta
\end{align*}
All time-simulations have been performed by solving the closed-loop system equation with initial conditions $\sigma(0)=[1~1]^T$. The robust control synthesis problems have been solved with $\bar \phi = \pi/6~[{\rm rad}]$, $\bar \Delta = \pi/4~[{\rm rad}]$ and $\alpha_u = 20$.
\begin{itemize}
\item {\bf No disturbance $(\delta = 0)$:} The VSC state feedback gain and the associated reaching time estimation follow from the optimal solution of problem (\ref{eq14}), that is
\begin{align*}
K_{VSC} \!=\! \left[ \!\! \begin{array}{rr} -5.6848 & 3.2821 \\ -3.2821 & -5.6848 \end{array} \!\! \right  ],~T_r(\sigma_0) \leq 0.4309
\end{align*}
The UVC state feedback gain and the associated reaching time estimation are calculated from the optimal solution of problem (\ref{eq27}) which, for $\rho = 4.0$, yields
\begin{align*}
K_{UVC} \!=\! \left[ \!\! \begin{array}{rr} -4.6108 & 2.6622 \\ -2.6617 & -4.6105 \end{array} \!\! \right ],~T_r(\sigma_0) \leq 0.3764
\end{align*}
As expected, without exogenous disturbance the magnitude of the gains and the closed-loop reaching time estimation are similar. We have determined by time-simulation that the true reaching time is $\approx 0.21$ and $\approx 0.37$, respectively showing that, in this example, the UVC is quite precise. The VSC appears to be more sensitive to parameter uncertainty.

\item {\bf Disturbance with $(\delta = 2.0)$:} The VSC and the UVC control strategies can now cope with the disturbance $f(t)$ given before. From the optimal solution of problem (\ref{eq14}) we obtain
\begin{align*}
K_{VSC} \!=\! \left[ \!\! \begin{array}{rr} -6.6295 & 3.8276 \\ -3.8276 & -6.6295 \end{array} \!\! \right  ],~T_r(\sigma_0) \leq 0.5860
\end{align*}
whereas, the optimal solution of (\ref{eq27}) for $\rho=3.0$ yields
\begin{align*}
K_{UVC} \!=\! \left[ \!\! \begin{array}{rr}  -5.9925 & 3.4599 \\ -3.4595 & -5.9923 \end{array} \!\! \right ],~T_r(\sigma_0) \leq 0.4893
\end{align*}
Observe the impact on the reaching time estimation due to the presence of the exogenous disturbance. The numerical determination of the true reaching time is a complicated task since the given disturbance $f(t)$ is not, in general, the one that produces the maximum reaching time among all feasible disturbances. This problem has been tackled in \cite{gero:2023b}. 
\end{itemize}
\begin{figure}[t]
	\begin{center}
	\includegraphics[width=0.4\textwidth]{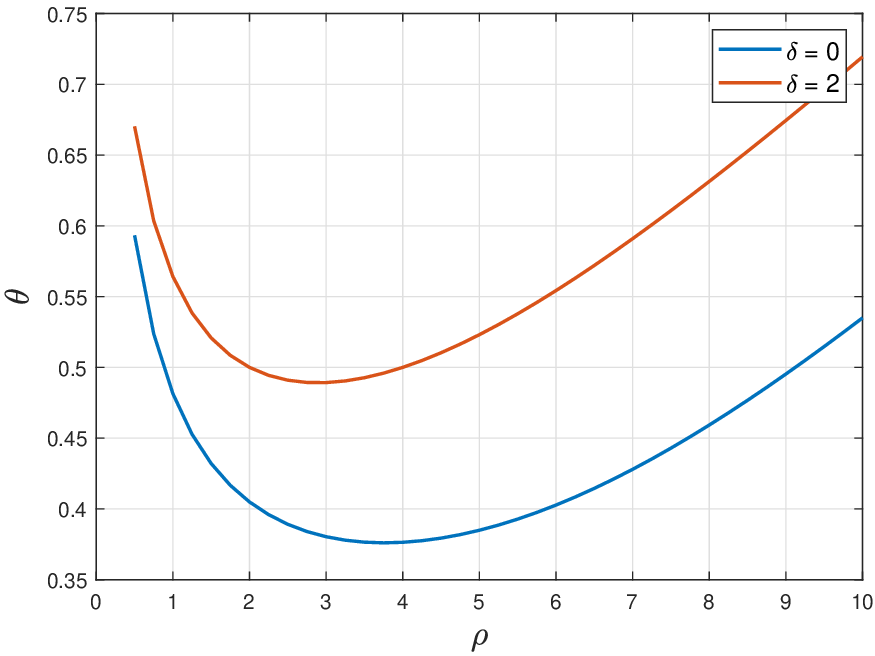}
		\caption{The plot of $\theta$ against $\rho$ for UVC design with the closed-loop system free of exogenous disturbance ($\delta =0$) and subject to exogenous norm bounded disturbance with ($\delta = 2.0$).}
		\label{Fig00}
	\end{center}
\end{figure}

It is interesting to observe the variation of the inner problem objective function in (\ref{eq27}) with respect to the scalar variable $\rho>0$. Figure \ref{Fig00} shows $\theta$ against $\rho$ for $\delta =0$ and $\delta=2$. In both case, they are unimodal with a clear minimum at $\rho\approx 4.0$ and $\rho \approx 3.0$, respectively. Of course, this behavior simplifies the outer minimization indicated in (\ref{eq27}) by the line search procedure. Unfortunately, such an occurrence may not be verified in different problems. 

\subsection{Underwater Remotely Operated Vehicle }

This is an over-actuated ROV (underwater Remotely Operated Vehicle) third order model inspired from a real ROV (LUMA)\footnote{More detail about the ROV can be found in \newline \url{https://insac.eesc.usp.br/luma/} (in Portuguese).}. The state variable is $\sigma = [v_x~v_y~\omega_z]^T \in \Rf^3$ where $v_x$ and $v_y$ are the velocities related to the body frame and $\omega_z$ is the angular velocity related to the $z$-axis. The ROV has four propellers responsible for the displacement of the body through the control variable $u \in \Rf^4$. Neglecting the Coriolis, drag and tether forces \cite{CCH:95}, a simplified ROV model with uncertain parameters developed in \cite{gero:2023b} can be expressed by (\ref{eq01}) with
\begin{align*}
B(g) = M^{-1} \Psi \Pi(g)
\end{align*}
where $M = {\rm diag}(m_0, m_0, I_z)$, with ROV mass $m_0 = 290~[{\rm kg}]$, moment of inertia $I_z = 23~[{\rm kg~ m^2}]$ and
\begin{align*}
\Psi = \left [ \begin{array}{rrrr} \psi_1 & \psi_1 & \psi_1 & \psi_1 \\ \psi_1 & -\psi_1 & -\psi_1 & \psi_1 \\ -\psi_2 & \psi_2 & -\psi_2 & \psi_2 \end{array} \right ]
\end{align*}
where $\psi_1 = \sqrt{2}/2$ and $\psi_2 = 0.35~[{\rm m}]$. The matrix $\Pi(g) = {\rm diag}(g_1, 1, g_3, 1)$ defines the actuator gains. It is assumed that $g_1, g_3 \in [1/2~1]$ are uncertain gains of the first and third control channels, while the other ones operate normally. It is immediate to determine $N=4$ extreme matrices $B_i \in \Rf^{3 \times 4}$ such that $B(g) \in \Bf = {\rm co}\{B_i\}_{i \in \Kn}$ for all feasible pairs $(g_1, g_3)$. Furthermore, we have imposed the initial condition $\sigma(0) = [1~1~\pi/4]^T$. Assuming that the system is disturbance free ($\delta = 0$) and setting $\alpha_u = 10^3$, both robust control synthesis problems have been solved.
\begin{itemize}
\item {\bf VSC design:} The optimal solution of problem (\ref{eq14}) provides the closed-loop reaching time estimation $T_r(\sigma_0) \leq 1.3037$ imposed by the state feedback gain
\begin{align*}
K_{VSC} = \left [ \begin{array}{rrr} 
  -363.6954 & -163.3043 & 90.6665 \\
 -164.1840 & 318.5337 & -193.6286  \\
 -363.6954 & 163.3043 &  90.6665  \\
 -164.1840 & -318.5337 & -193.6286 \end{array} \right ]
\end{align*}
Figure \ref{Fig01} shows the time simulation of the closed-loop system variable $\sigma(t)$ and the mark ''$\bullet$'' indicating the theoretical reaching time estimation. The precision is heavily influenced by the parameter uncertainty.
\begin{figure}[t]
	\begin{center}
		\includegraphics[width=0.4\textwidth]{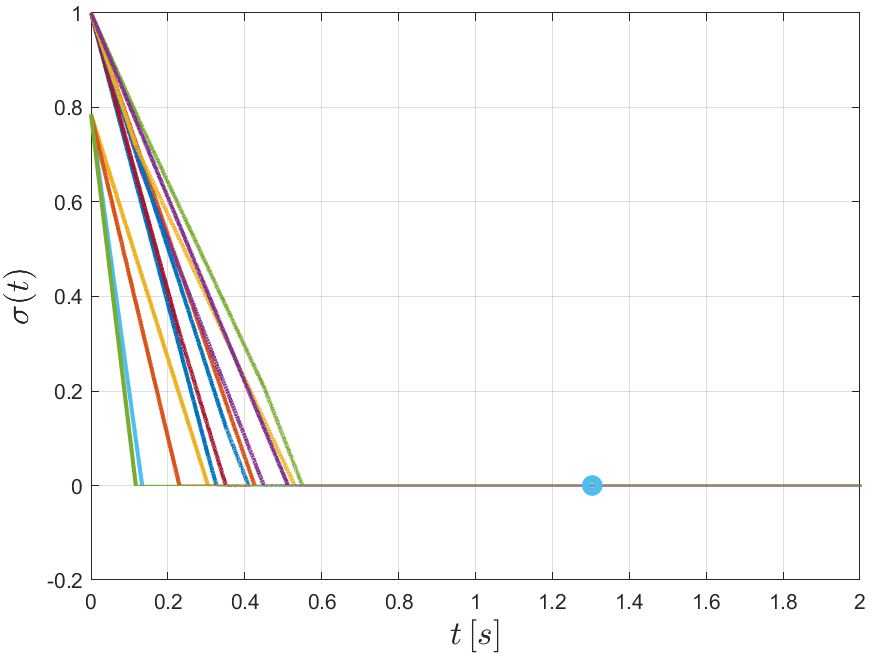}
		\caption{Closed-loop system time evolution of the state variables of the over-actuated ROV under VSC action with input matrix corresponding to each vertex of the convex bounded uncertainty domain $\mathbb{B}$ and free of external disturbance.}
		\label{Fig01}
	\end{center}
\end{figure}
\item {\bf UVC design:} The optimal solution of the problem (\ref{eq27}) obtained for $\rho = 2.0$ provides the reaching time estimation $T_r(\sigma_0) \leq 0.7570$  and the associated state feedback gain 
\begin{align*}
K_{UVC} & = \left [ \begin{array}{rrr} 
  -383.9188 & -102.9356 & 107.3066  \\
 -167.5785 &  308.5779 & -179.9052 \\
 -375.5741 & 244.7841  & 75.9815  \\
 -193.1726 & -368.4072 & -96.1320 \end{array} \right ]
\end{align*}
The time simulation of the closed-loop system given in Figure \ref{Fig02} shows the state variable $\sigma(t)$ and the mark ''$\bullet$'' indicating the theoretical reaching time estimation. Once again, the precision of the UVC is remarkable.
\begin{figure}[t]
	\begin{center}
		\includegraphics[width=0.40\textwidth]{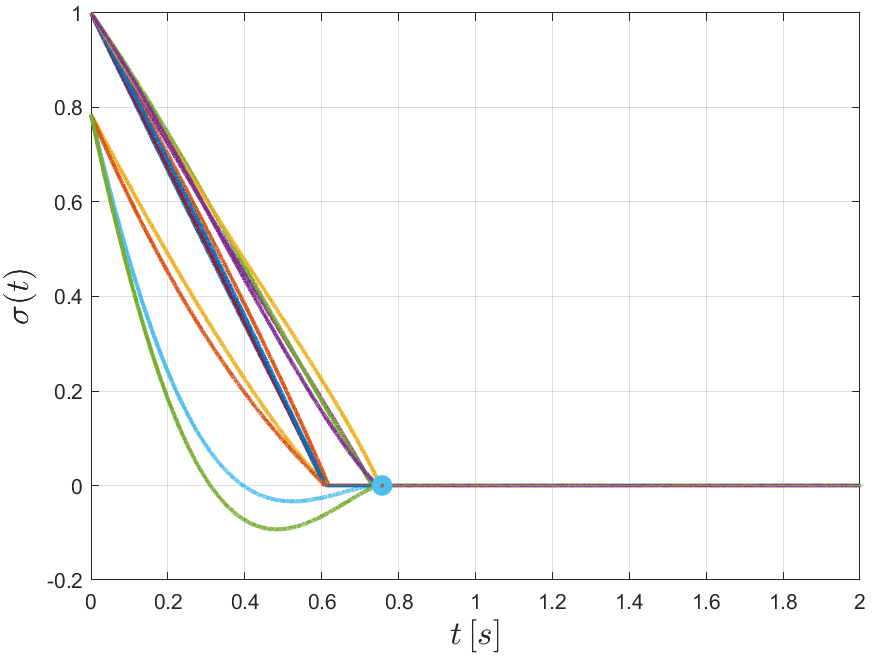}
		\caption{Closed-loop system time evolution of the state variables of the over-actuated ROV under UVC action with input matrix corresponding to each vertex of the convex bounded uncertainty domain $\mathbb{B}$ and free of external disturbance.}
		\label{Fig02}
	\end{center}
\end{figure}
\end{itemize}

\section{Conclusions}

In this paper, we presented new results on the analysis and synthesis of minimum reaching time for two well-known classes of control strategies, assuring finite time convergence of the closed-loop system towards the origin. Each optimal multivariable state feedback gain is robust since parameter uncertainty and norm bounded exogenous disturbance are taken into account. All design conditions are expressed by LMIs, allowing the determination of the optimal control law, with no difficulty, by the machinery available in the literature. The optimal Variable Structure Control was designed with the help of a Persidiskii type Lyapunov function and the examples showed that, in general, its reaching time estimation is sensitive to parameter uncertainty. Besides, the optimal Unit Vector Control law was determined by the adoption of a quadratic Lyapunov function and, for the same examples, it provided very precise reaching time estimations. A natural follow-up of this paper is the development of a new class of Lyapunov function for VSC design by trying to relax the diagonal structure of the involved matrices imposed by the Persidiskii stability theory and the theoretical verification of the surprising quality of the UVC reaching-time estimation produced by a pure quadratic Lyapunov function.

\balance

\section*{Acknowledgment}

The authors would like to thank the National Council for Scientific and Technological Development (CNPq/Brazil) and the Coordination for the Improvement of Higher Education Personnel (CAPES/Brazil) for the partial financial support along the development of this research.

\bibliographystyle{ieeetr}
\bibliography{gscar, MSTA_LMI,BIB,vsmrac, VSS2022_Chattering}

\end{document}